\documentclass[11pt]{article}
\usepackage[a4paper,margin=1in]{geometry}
\usepackage[T1]{fontenc} 
\usepackage{microtype} 
\usepackage{hyperref}
\hypersetup{colorlinks=true,linkcolor=[rgb]{0.75,0,0},citecolor=[rgb]{0,0,0.75}}
\usepackage{graphicx}
\graphicspath{{./graphics/}}
\usepackage{amsmath}
\usepackage{amssymb}
\usepackage{color}
\usepackage{amsthm}
\usepackage{thmtools}
\declaretheorem[name=Theorem,numberwithin=section]{theorem}
\declaretheorem[name=Lemma,sibling=theorem]{lemma}

\newcommand{\graphicsScale}{1.}
\newcommand{\Figure}{Fig.}

\newcommand{\OO}{\mathcal{O}} 
\newcommand{\dd}{\mathbf{d}} 
\newcommand{\D}{\mathbf{D}} 
\newcommand{\DTSP}{\D_\mathtt{TSP}} 
\newcommand{\DRB}{\D_\mathtt{RB}} 
\newcommand{\DMM}{\D_\mathtt{MM}} 
\newcommand{\DGG}{\D_\mathtt{G}} 
\newcommand{\flip}[2]{#1 \twoheadrightarrow #2} 

\title{On the Longest Flip Sequence to Untangle Segments in the Plane\texorpdfstring{\thanks{A version of this paper appears in Walcom'23. This work is supported by the French ANR PRC grant ADDS (ANR-19-CE48-0005).}}{}}

\author{
  Guilherme D. da Fonseca\thanks{Aix-Marseille Université and LIS, France. guilherme.fonseca@lis-lab.fr} 
  \and Yan Gerard\thanks{Université Clermont Auvergne and LIMOS, France. \{yan.gerard, bastien.rivier\}@uca.fr} 
  \and Bastien Rivier\textsuperscript{$\ddagger$} 
}

\begin{document}
\maketitle

\begin{abstract}
  A set of segments in the plane may form a Euclidean TSP tour or a matching, among others. Optimal TSP tours as well as minimum weight perfect matchings have no crossing segments, but several heuristics and approximation algorithms may produce solutions with crossings. To improve such solutions, we can successively apply a flip operation that replaces a pair of crossing segments by non-crossing ones. This paper considers the maximum number $\D(n)$ of flips performed on $n$ segments. First, we present reductions relating $\D(n)$ for different sets of segments (TSP tours, monochromatic matchings, red-blue matchings, and multigraphs). Second, we show that if all except $t$ points are in convex position, then $\D(n) = \OO(tn^2)$, providing a smooth transition between the convex $\OO(n^2)$ bound and the general $\OO(n^3)$ bound. Last, we show that if instead of counting the total number of flips, we only count the number of distinct flips, then the cubic upper bound improves to $\OO(n^{8/3})$. 
\end{abstract}

\section{Introduction}
\label{sec:intro}

In the Euclidean Travelling Salesman Problem (TSP), we are given a set $P$ of $n$ points in the plane and the goal is to produce a closed tour connecting all points of minimum total Euclidean length. The TSP problem, both in the Euclidean and in the more general graph versions, is one of the most studied NP-hard optimization problems, with several approximation algorithms, as well as powerful heuristics (see for example~\cite{ABCC11,Dav10,GuPu06}). Multiple PTAS are known for the Euclidean version~\cite{Aro96,Mit99,RaSm98}, in contrast to the general graph version that unlikely admits a PTAS~\cite{ChC19}. It is well known that the optimal solution for the Euclidean TSP is a simple polygon, i.e., has no crossing segments, and in some situation a crossing-free solution is necessary~\cite{BuKi22}. However, most approximation algorithms (including Christofides and the PTAS), as well as a variety of simple heuristics (nearest neighbor, greedy, and insertion, among others) may produce solutions with pairs of crossing segments. In practice, these algorithms may be supplemented with a local search phase, in which crossings are removed by iteratively modifying the solution.

\begin{figure}[htb]
  \centering
  \includegraphics[scale=\graphicsScale]{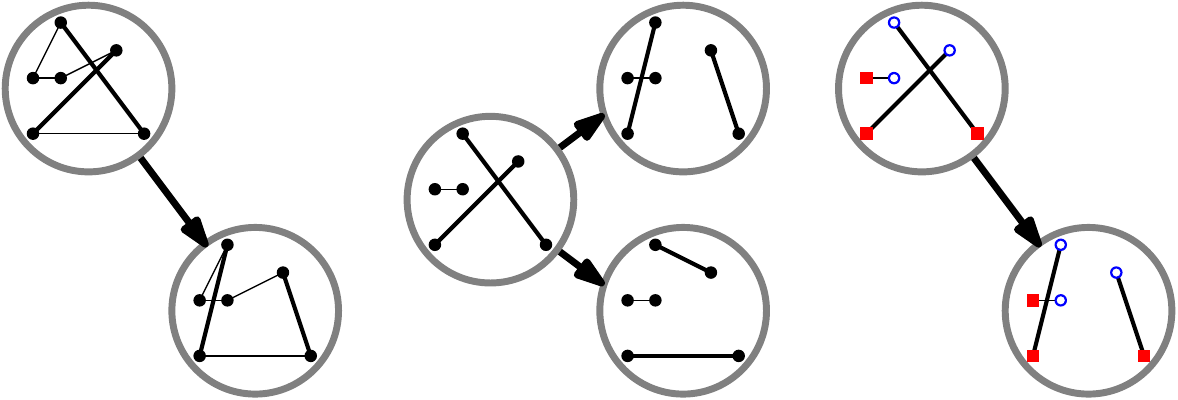}\\
  (a) \hspace{0.3\textwidth}(b)  \hspace{0.3\textwidth}(c)
  \caption{Examples of flips in a (a) TSP tour, (b) monochromatic matching, and (c) red-blue matching.}
  \label{fig:flip}
\end{figure}

Given a Euclidean TSP tour, a \emph{flip} is an operation that removes a pair of crossing segments and adds a new pair of segments preserving a tour (\Figure~\ref{fig:flip}(a)). If we want to find a tour without crossing segments starting from an arbitrary tour, it suffices to find a crossing, perform a flip, and repeat until there are no crossings. It is easy to see that the process will eventually finish, as the length of the tour can only decrease when we perform a flip. Since a flip may create several new crossings, it is not obvious how to bound the number of flips performed until a crossing-free solution is obtained. Let $\DTSP(n)$ denote the maximum number of flips successively performed on a TSP tour with $n$ segments. An upper bound of $\DTSP(n) = \OO(n^3)$ is proved in~\cite{VLe81}, while the best lower bound known is $\DTSP(n) = \Omega(n^2)$.
In contrast, if the points $P$ are in convex position, then tight bounds of $\Theta(n^2)$ are easy to prove.

In this paper, we show that we can consider a conceptually simpler problem of flips in matchings (instead of Hamiltonian cycles), in order to bound the number of flips to both problems. Next, we describe this \emph{monochromatic matching version}.

Consider a set of $n$ line segments in the plane defining a matching $M$ on a set $P$ of $2n$  points. In this case, a \emph{flip} replaces a pair of crossing segments by a pair of non-crossing ones using the same four endpoints (\Figure~\ref{fig:flip}(b)). Notice that, in contrast to the TSP version, one of two possible pairs of non-crossing segments is added.
As previously, let $\DMM(n)$ denote the maximum number of flips successively performed on a monochromatic matching with $n$ segments.
In Section~\ref{sec:reduction}, we show that $\DMM(n) \leq \DTSP(6n) \leq \DMM(6n)$, hence it suffices to prove asymptotic bounds for $\DMM(n)$ in order to bound $\DTSP(n)$. 

A third and last version of the problem that we consider is the \emph{red-blue matching version}, in which the set $P$ is partitioned into two sets of $n$ points each called \emph{red} and \emph{blue}, with segments only connecting points of different colors (\Figure~\ref{fig:flip}(c)). Let $\DRB(n)$ denote the analogous maximum number of flips successively performed on a red-blue matching with $n$ segments. The red-blue matching version has been thoroughly studied~\cite{BMS19,DDFGR22}. In Section~\ref{sec:reduction}, we also show that $\DMM(n) \leq \DRB(2n) \leq \DMM(2n)$ and, as a consequence, asymptotic bounds for the monochromatic matching version also extend to the red-blue matching version. We use the notation $\D(n)$ for bounds that hold in all three versions.

For all the aforementioned versions, special cases arise when we impose some constraint on the location of the points $P$. In the \emph{convex case}, $P$ is in \emph{convex} position. Then it is known that, for all three versions, $\D(n) = \Theta(n^2)$~\cite{BMS19,BoM16}. This tight bound contrasts with the gap for the general case bounds.

For all three versions in the convex case, $\D(n) \leq \binom{n}{2}$ as the number of crossings decreases at each flip.
The authors have recently shown that without convexity $\DRB(n) \geq 1.5 \binom{n}{2} - \frac{n}{4}$~\cite{DDFGR22}, which is higher than the convex bound. 
A major open problem conjectured in~\cite{BoM16} is to determine if the non-convex bounds are $\Theta(n^2)$ as the convex bounds.
Unfortunately, the best upper bound known for the non-convex case remains $\D(n) = \OO(n^3)$~\cite{VLe81} since 1981\footnote{While the paper considers only the TSP version, the proof of the upper bound also works for the matching versions, as shown in~\cite{BoM16}.}, despite recent work on this specific problem~\cite{BMS19,BoM16,DDFGR22}. The best lower bound known is $\D(n) = \Omega(n^2)$~\cite{BoM16,DDFGR22}.

The argument for the convex case bound of $\D(n) \leq \binom{n}{2}$ breaks down even if all but one point are in convex position, as the number of crossings may not decrease. In Section~\ref{sec:ConvexToGeneral}, we present a smooth transition between the convex and the non-convex cases. We show that, in all versions, if there are $t$ points anywhere and the remaining points are in convex position with a total of $n$ segments, then the maximum number of flips is $\OO(tn^2)$. 

Finally, in Section~\ref{sec:distinct}, we use a balancing argument similar to the one of Erdös et al.~\cite{ELS73} to show that if, instead of counting the number of flips, we count the number of distinct flips (two flips are the same if they change the same set of four segments), then we get a bound of $\OO(n^{8/3})$.

\subsection{Related Reconfiguration Problems}

Combinatorial reconfiguration studies the step-by-step transition from one solution to another, for a given combinatorial problem. Many reconfiguration problems are presented in~\cite{Heu13}.

It may be tempting to use an alternative definition for a flip in order to remove crossings and reduce the length of a TSP tour. The \emph{2OPT flip} is not restricted to crossing segments, as long as it decreases the Euclidean length of the tour. However, the number of 2OPT flips performed may be exponential~\cite{ERV14}.

Another important parameter $\dd(n)$ is the \emph{minimum} number of flips needed to remove crossings from any set of $n$ segments. When the points are in convex position, then it is known that $\dd(n) = \Theta(n)$ in all three versions~\cite{BMS19,DDFGR22,OdW07,WCL09}. If the red points are on a line, then  $\dd_\texttt{RB}(n) = \OO(n^2)$~\cite{BMS19,DDFGR22}. In the monochromatic matching version, if we can choose which pair of segments to add in a flip, then $\dd_\texttt{MM}(n) = \OO(n^2)$~\cite{BoM16}. For all remaining cases, the best bounds known are $\dd(n) = \Omega(n)$ and $\dd(n) = \OO(n^3)$.

It is also possible to relax the flip definition to all operations that replace two segments by two others with the same four endpoints, whether they cross or not~\cite{BeI08,BeI17,BBH19,BJ20,EKM13,Wil99}. This definition has also been generalized to multigraphs with the same degree sequence~\cite{Hak62,Hak63,phdJof}. 

In the context of triangulations, a flip consists of removing one segment and adding another one while preserving a triangulation. 
Reconfiguration problems for triangulations are also widely studied~\cite{AMP15,HNU99,Law72,LuP15,NiN18,Pil14}.

\subsection{Definitions}
\label{sec:definitions}

Consider a set of points $P$.
We say that two segments $s_1,s_2 \in \binom{P}{2}$ \emph{cross} if they intersect in exactly one point which is not an endpoint of either $s_1$ or $s_2$. Furthermore, a line $\ell$ and a segment $s$ \emph{cross} if they intersect in exactly one point that is not an endpoint of $s$.

Let $s_1,s_1',s_2,s_2'$ be four segments of $\binom{P}{2}$ forming a cycle with $s_1,s_2$ crossing.
We define a \emph{flip} $f = \flip{s_1,s_2}{s_1',s_2'}$ as the function that maps any set (or multiset) of segments $M$ containing the two crossing segments $s_1$ and $s_2$ to $f(M) = M \cup \{s_1',s_2'\} \setminus \{s_1,s_2\}$ provided that $f(M)$ satisfies the property required by the version of the problem in question (being a monochromatic matching, a red-blue matching, a TSP tour...).
This leads to the most general version of the problem, called the \emph{multigraph} version. We note that a flip preserves the degree of every point. However, a flip may not preserve the multiplicity of a segment, which is why, in certain versions, we must consider multisets and multigraphs and not just sets and graphs.

A \emph{flip sequence} of \emph{length} $m$ is a sequence of flips $f_1,\ldots,f_m$ with a corresponding sequence of (multi-)sets of segments $M_0,\ldots,M_m$ such that $M_i = f_i(M_{i-1})$ for $i=1,\ldots,m$.
Unless mentioned otherwise, we assume general position for the points in $P$ (no three collinear points).

Given a property $\Pi$ over a multiset of $n$ line segments, we define $\D_\Pi(n)$ as the maximum length of a flip sequence such that every multiset of $n$ segments in the sequence satisfies property $\Pi$. We consider the following properties $\Pi$: $\texttt{TSP}$ for Hamiltonian cycle, $\texttt{RB}$ for red-blue matching, $\texttt{MM}$ for monochromatic matching, and $\texttt{G}$ for multigraph. Notice that if a property $\Pi$ is stronger than a property $\Pi'$, then $\D_\Pi(n) \leq \D_{\Pi'}(n)$.

\section{Reductions}
\label{sec:reduction}

In this section, we provide a series of inequalities relating the different versions of $\D(n)$. We show that all different versions of $\D(n)$ have the same asymptotic behavior.

\begin{theorem}
  \label{thm:reduction}
  For all positive integer $n$, we have the following relations
  \begin{align}
    \DMM(n) & = \DGG(n), \label{eq:0}\\
    2\DMM(n) & \leq \DRB(2n)~ \leq \DMM(2n), \label{eq:1}\\
    2\DRB(n) & \leq \DTSP(3n) \leq \DMM(3n). \label{eq:2}
  \end{align}
\end{theorem}
\begin{proof}

  Equality~\ref{eq:0} can be rewritten $\DGG(n) \leq \DMM(n) \leq  \DGG(n)$. Hence, we have to prove six inequalities. The right-side inequalities are immediate, since the left-side property is stronger than the right-side property (using \texttt{G} instead of \texttt{MM} for inequality~\ref{eq:2}). 

  The proofs of the remaining inequalities follow the same structure: given a flip sequence of the left-side version, we build a flip sequence of the right-side version, having similar length and number of points.

  We first prove the inequality $\DGG(n) \leq \DMM(n)$. 
  A point of degree $\delta$ larger than $1$ can be replicated as $\delta$ points that are arbitrarily close to each other in order to produce a matching of $2n$ points.
  This replication preserves the crossing pairs of segments (possibly creating new crossings). Thus, for any flip sequence in the multigraph version, there exists a flip sequence in the monochromatic matching version of equal length, yielding $\DGG(n) \leq \DMM(n)$.

  \begin{figure}[!ht]
    \centering
    (a) \includegraphics[scale=\graphicsScale]{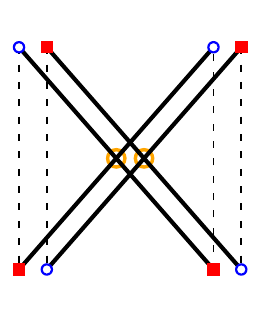}
    \hspace{.2\textwidth}
    (b) \includegraphics[scale=\graphicsScale]{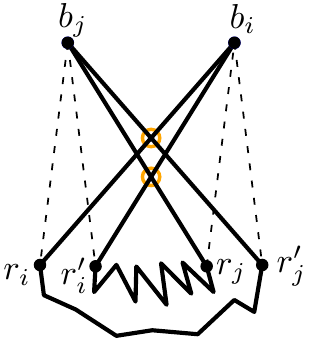}
    \caption{(a) Two red-blue flips to simulate a monochromatic flip. (b) Two TSP flips to simulate a monochromatic flip.}
    \label{fig:simulateFlip}
  \end{figure}

  The left inequality of~\eqref{eq:1} is obtained by duplicating the monochromatic points of the matching $M$ into two arbitrarily close points, one red and the other blue. Then each segment of $M$ is also duplicated into two red-blue segments. We obtain a bichromatic matching $M'$ with $2n$ segments. 
  A crossing in $M$ corresponds to four crossings in $M'$. Flipping this crossing in $M$ amounts to choose which of the two possible pairs of segments replaces the crossing pair. It is simulated by flipping the two crossings in $M'$ such that the resulting pair of double segments corresponds to the resulting pair of segments of the initial flip. These two crossings always exist and it is always possible to flip them one after the other as they involve disjoint pairs of segments.
  \Figure~\ref{fig:simulateFlip}(a) shows this construction.
  A sequence of $m$ flips on $M$ provides a sequence of $2m$ flips on $M'$.
  Hence, $2 \DMM(n) \leq \DRB(2n)$.

  To prove the left inequality of~\eqref{eq:2}, we start from a red-blue matching $M$ with $2n$ points and $n$ segments and build a tour $T$ with $3n$ points and $3n$ segments. We then show that the flip sequence of length $m$ on $M$ provides a flip sequence of length $2m$ on $T$. We build $T$ in the following way. Given a red-blue segment $rb \in M$, the red point $r$ is duplicated in two arbitrarily close points $r$ and $r'$ which are adjacent to $b$ in $T$. We still need to connect the points $r$ and $r'$ in order to obtain a tour $T$. 
  We define $T$ as the tour
  $r_1,b_1,r'_1,\ldots, r_i,b_i,r'_i,\ldots,r_n,b_n,r'_n\ldots$
  where $r_i$ is matched to $b_i$ in $M$ (\Figure~\ref{fig:simulateFlip}(b)).

  We now show that a flip sequence of $M$ with length $m$ provides a flip sequence of $T$ with length $2m$. For a flip $\flip{r_ib_i,r_jb_j}{r_ib_j,r_jb_i}$ on $M$, we perform two successive flips $\flip{r_ib_i,r'_jb_j}{r_ib_j,r'_jb_i}$ and $\flip{r'_ib_i,r_jb_j}{r'_ib_j,r_jb_i}$ on $T$.

  The tour then becomes
  $
  r_1,b_1,r'_1,\ldots, r_i,b_j,r'_i,\ldots, r_j,b_i,r'_j,\ldots,r_n,b_n,r'_n,\ldots
  $
  on which we can apply the next flips in the same way. Hence, $2 \DRB(n) \leq \DTSP(3n)$, concluding the proof. 
\end{proof}

\section{Near Convex Sets}
\label{sec:ConvexToGeneral}

In this section, we bridge the gap between the $\OO(n^2)$ bound on the length of flip sequences for a set $P$ of points in convex position and the $\OO(n^3)$ bound for $P$ in general position.
We prove the following theorem in the monochromatic matching version; the translations to the other versions follows from the reductions from Section~\ref{sec:reduction}, noticing that, when the points have degree $\OO(1)$, all reductions preserve the number of points in non-convex position up to constant factors.

\begin{theorem}
  \label{thm:ConvexToGeneralBis}
  In the monochromatic matching version with $n$ segments, if all except $t$ points of $P$ are in convex position, then the length of a flip sequence is $\OO(t n^2)$.
\end{theorem}
\begin{proof}
  The proof strategy is to combine the potential $\Phi_X$ used in the convex case with the potential $\Phi_L$ used in the general case.
  Given a matching $M$, the potential $\Phi_X(M)$ is defined as the number of crossing pairs of segments in $M$. Since there are $n$ segments in $M$, $\Phi_X(M) \leq \binom{n}{2} = \OO(n^2)$.
  Unfortunately, with points in non-convex position, a flip $f$ might \emph{increase} $\Phi_X$, i.e. $\Phi_X(f(M)) \geq \Phi_X(M)$ (as shown in \Figure~\ref{fig:flip}).

  The potential $\Phi_L$ is derived from the line potential introduced in~\cite{VLe81} but instead of using the set of all the $\OO(n^2)$ lines through two points of $P$, we use a subset of $\OO(tn)$ lines in order to take into account that only $t$ points are in non-convex position. 
  More precisely, let the potential $\Phi_\ell(M)$ of a line $\ell$ be the number of segments of $M$ crossing $\ell$. Note that $\Phi_\ell(M) \leq n$. The potential $\Phi_L(M)$ is then defined as follows: $\Phi_L(M) = \sum_{\ell \in L} \Phi_\ell(M)$.

  We now define the set of lines $L$ as the union of $L_1$ and $L_2$, defined hereafter. Let $C$ be the subset containing the $2n-t$ points of $P$ which are in convex position. Let $L_1$ be the set of the $\OO(tn)$ lines through two points of $P$, at least one of which is not in $C$. Let $L_2$ be the set of the $\OO(n)$ lines through two points of $C$ which are consecutive on the convex hull boundary of $C$.

  Let the potential $\Phi(M) = \Phi_X(M) + \Phi_L(M)$.
  We have the following bounds: $0 \leq \Phi(M) \leq \OO(tn^2)$.
  To complete the proof of Theorem~\ref{thm:ConvexToGeneralBis}, we show that any flip decreases $\Phi$ by at least $1$ unit.

  We consider an arbitrary flip $f=\flip{p_1p_3,p_2p_4}{p_1p_4,p_2p_3}$. Let $p_x$ be the point of intersection of $p_1p_3$ and $p_2p_4$.
  It is shown in~\cite{DDFGR22,VLe81} that $f$ never increases the potential $\Phi_\ell$ of a line $\ell$. More precisely, we have the following three cases:
  \begin{itemize}
  \item The potential $\Phi_\ell$ decreases by $1$ unit if the line $\ell$ separates the final segments $p_1p_4$ and $p_2p_3$ and exactly one of the four flipped points belongs to $\ell$. We call these lines \emph{$f$-critical} (\Figure~\ref{fig:criticalPotentialDrop}(a)).
  \item The potential $\Phi_\ell$ decreases by $2$ units if the line $\ell$ strictly separates the final segments $p_1p_4$ and $p_2p_3$. We call these lines \emph{$f$-dropping} (\Figure~\ref{fig:criticalPotentialDrop}(b)). 
  \item The potential $\Phi_\ell$ remains stable in the remaining cases.
  \end{itemize}
  Notice that, if a point $q$ lies in the triangle $p_1p_xp_4$, then the two lines $qp_1$ and $qp_4$ are $f$-critical (\Figure~\ref{fig:criticalPotentialDrop}(a)). 

  \begin{figure}[!ht]
    \centering
    (a) \includegraphics[page=1,scale=\graphicsScale]{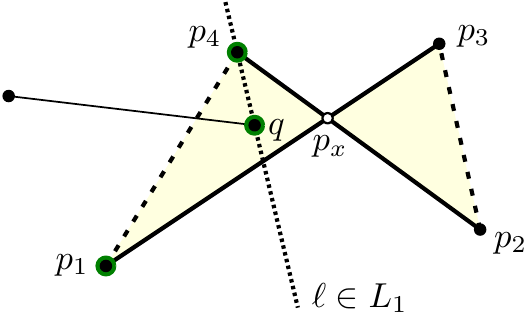}\hfill
    (b) \includegraphics[page=2,scale=\graphicsScale]{criticalPotentialDrop}
    \caption{(a) An $f$-critical line $\ell$ for a flip $f=\flip{p_1p_3,p_2p_4}{p_1p_4,p_2p_3}$. This situation corresponds to case (2a) with $\ell \in L_1$. (b) An $f$-dropping line $\ell$. This situation corresponds to case (2b) with $\ell \in L_2$.} 
    \label{fig:criticalPotentialDrop}
  \end{figure}

  To prove that $\Phi$ decreases, we have the following two cases.

  \textbf{Case 1.} If $\Phi_X$ decreases, as the other term $\Phi_L$ does not increase, then their sum $\Phi$ decreases as desired. 

  \textbf{Case 2.} If not, then $\Phi_X$ increases by an integer $k$ with $0 \leq k \leq n-1$, and we know that there are $k+1$ new crossings after the flip $f$. Each new crossing involves a distinct segment with one endpoint, say $q_i$ ($0 \leq i \leq k$), inside the non-simple polygon $p_1,p_4,p_2,p_3$ (\Figure~\ref{fig:criticalPotentialDrop}). 
  Next, we show that each point $q \in \{q_0,\ldots,q_k\}$ maps to a distinct line in $L$ which is either $f$-dropping or $f$-critical, thus proving that the potential $\Phi_L$ decreases by at least $k+1$ units.

  We assume without loss of generality that $q$ lies in the triangle $p_1p_xp_4$.
  We consider the two following cases.

  \textbf{Case 2a.} If at least one among the points $q, p_1, p_4$ is not in $C$, then either $qp_1$ or $qp_4$ is an $f$-critical line $\ell \in L_1$ (\Figure~\ref{fig:criticalPotentialDrop}(a)).

  \textbf{Case 2b.} If not, then $q, p_1, p_4$ are all in $C$, and the two lines through $q$ in $L_2$ are both either $f$-dropping (the line $\ell$ in \Figure~\ref{fig:criticalPotentialDrop}(b)) or $f$-critical (the line $qp_4$ in \Figure~\ref{fig:criticalPotentialDrop}(b)).
  Consequently, there are more lines $\ell \in L_2$ that are either $f$-dropping or $f$-critical than there are such points $q \in C$ in the triangle $p_1p_xp_4$, and the theorem follows.
\end{proof}

\section{Distinct Flips}
\label{sec:distinct}

In this section, we prove the following theorem in the monochromatic matching version, yet, the proof can easily be adapted to the other versions. We remark that two flips are considered \emph{distinct} if the sets of four segments in the flips are different.

\begin{theorem}
  \label{thm:distinct}
  In all versions with $n$ segments, the number of \emph{distinct} flips in any flip sequence is $\OO(n^{{8}/{3}})$.
\end{theorem}

The proof of Theorem~\ref{thm:distinct} is based on a balancing argument from~\cite{ELS73} and is decomposed into two lemmas that consider a flip $f$ and two matchings $M$ and $M' = f(M)$. Similarly to~\cite{VLe81}, let $L$ be the set of lines defined by all pairs of points in $\binom{P}{2}$. For a line $\ell \in L$, let $\Phi_\ell(M)$ be the number of segments of $M$ crossed by $\ell$ and $\Phi_L(M)=\sum_{\ell \in L} \Phi_\ell(M)$. Notice that $\Phi(M) - \Phi(M')$ depends only on the flip $f$ and not on $M$ or $M'$. The following lemma follows immediately from the fact that $\Phi_L(M)$ takes integer values between $0$ and $\OO(n^3)$\cite{VLe81}.

\begin{lemma}
  \label{lem:bigDrops}
  For any integer $k$, the number of flips $f$ in a flip sequence with $\Phi(M) - \Phi(M') \geq k$ is $\OO(n^3/k)$.
\end{lemma}

Lemma~\ref{lem:bigDrops} bounds the number of flips (distinct or not) that produce a large potential drop in a flip sequence. Next, we bound the number of distinct flips that produce a small potential drop. The bound considers all possible flips on a fixed set of points and does not depend on a particular flip sequence.

\begin{lemma}
  \label{lem:smallDrops}
  For any integer $k$, the number of \emph{distinct} flips $f$ with $\Phi(M) - \Phi(M') < k$ is $\OO(n^2k^2)$.
\end{lemma}
\begin{proof}
  Let $F$ be the set of flips with $\Phi(M) - \Phi(M') < k$ where $M'=f(M)$. We need to show that $|F| = \OO(n^2k^2)$. Consider a flip $f=\flip{p_1p_3,p_2p_4}{p_1p_4,p_2p_3}$ in $F$. Next, we show that there are at most $4k^2$ such flips with a fixed final segment $p_1p_4$. Since there are $\OO(n^2)$ possible values for $p_1p_4$, the lemma follows. We show only that there are at most $2k$ possible values for $p_3$. The proof that there are at most $2k$ possible values for $p_2$ is analogous.

  \begin{figure}[!ht]
    \centering
    \includegraphics[scale=\graphicsScale]{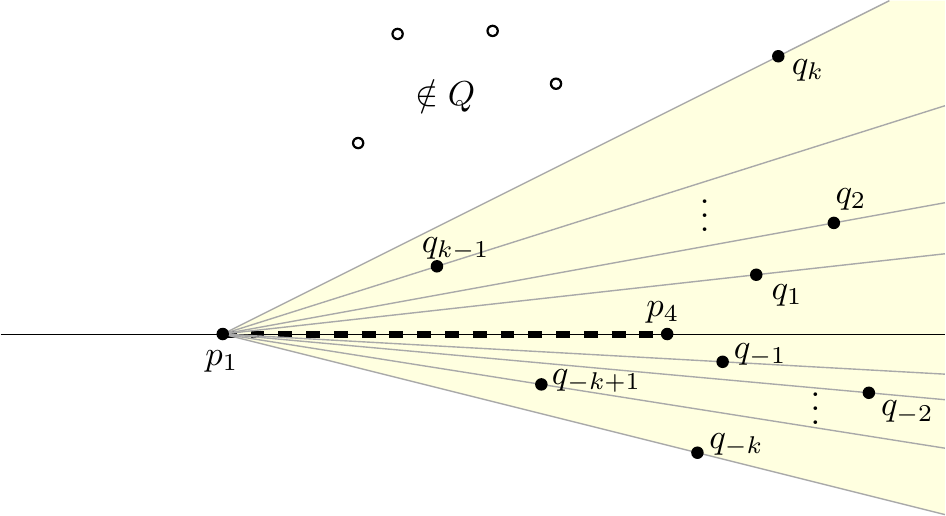}
    \caption{Illustration for the proof of Lemma~\ref{lem:smallDrops}.}
    \label{fig:smallDrops}
  \end{figure}

  We sweep the points in $P \setminus \{p_4\}$ by angle from the ray $p_1p_4$. As shown in Figure~\ref{fig:smallDrops}, let $q_1,\ldots,q_k$ be the first $k$ points produced by this sweep in one direction, $q_{-1}\ldots,q_{-k}$ in the other direction and $Q=\{q_{-k}\ldots,q_{-1},q_1,\ldots,q_k\}$. To conclude the proof, we show that $p_3$ must be in $Q$. 
  Suppose $p_3 \notin Q$ for the sake of a contradiction and assume without loss of generality that $p_3$ is on the side of $q_i$ with positive $i$. Then, consider the lines $L' = \{p_1q_1,\ldots,p_1q_k\}$.
  Notice that $L' \subseteq L$, $|L'| = k$, and for each $\ell \in L'$ we have $\Phi_\ell(M) > \Phi_\ell(M')$, which contradicts the hypothesis that $\Phi(M) - \Phi(M') < k$. 
\end{proof}

Theorem~\ref{thm:distinct} is a consequence of Lemmas~\ref{lem:bigDrops} and~\ref{lem:smallDrops} with $k = n^{{1}/{3}}$.

\section{Conclusion and Open Problems}

In Section~\ref{sec:reduction}, we showed a relationship among several different bounds on the maximum number of flips for a set of $n$ line segments. This result shows how upper bounds to the monochromatic matching version can be easily transferred to different versions. But they can also be applied to transfer lower bounds among different versions. For example, the lower bound of $\frac{3}{2}\binom{n}{2} - \frac{n}{4}$ for the red-blue matching case~\cite{DDFGR22} implies a lower bound of $\frac{1}{3} \binom{n}{2} - \frac{n}{2}$

for TSP. It is not clear if the constants in the TSP lower bound may be improved, perhaps by a more direct approach (or perhaps the lower bounds are not even asymptotically tight). However, we showed that all these versions are related by constant factors.

We can also use the results from Section~\ref{sec:reduction} to convert the bounds from Section~\ref{sec:ConvexToGeneral} to the general multigraph version, and hence to spanning trees and other types of graphs. In this case, the number $t$ of points in non-convex position needs to be replaced by the sum of the degrees of the points in non-convex position. If the graphs are dense, or have non-constant degree, we do not know of any non-trivial lower bounds or better upper bounds.

Another key property that we did not consider in this paper is the length $\dd(n)$ of the shortest flip sequence to untangle any set of $n$ segments. In general, we only know that $\dd(n) = \Omega(n)$ and $\dd(n) \leq \D(n) = \OO(n^3)$, for all versions. Whether similar reductions are possible is an elusive question. Furthermore, if the $n$ points are in convex position, then $\dd(n) = \Theta(n)$ for all versions. It is unclear if a transition in the case of $t$ points in non-convex position is possible.

The result of Theorem~\ref{thm:distinct} is based on the $\OO(n^2k^2)$ bound from Lemma~\ref{lem:smallDrops}. A better analysis of the dual arrangement could potentially improve this bound, perhaps to $\OO(n^2k)$.

The main open question, though, is whether the $\OO(n^3)$ bound to both $\D(n)$ and $\dd(n)$ can be improved for points in general position. The $\OO(n^{8/3})$ bound on the number of distinct flips presented in Section~\ref{sec:distinct} is a hopeful step in this direction, at least for $\dd(n)$. We were not able to find a set of line segments that requires the same pair of segments to be flipped twice in order to be untangled.

\bibliography{ref}

\begin{thebibliography}{10}

\bibitem{AMP15}
Oswin Aichholzer, Wolfgang Mulzer, and Alexander Pilz.
\newblock Flip distance between triangulations of a simple polygon is
  {NP}-complete.
\newblock {\em Discrete \& Computational Geometry}, 54(2):368--389, 2015.

\bibitem{ABCC11}
David~L Applegate, Robert~E Bixby, Va{\v{s}}ek Chv{\'a}tal, and William~J Cook.
\newblock The traveling salesman problem.
\newblock In {\em The Traveling Salesman Problem}. Princeton university press,
  2011.

\bibitem{Aro96}
Sanjeev Arora.
\newblock Polynomial time approximation schemes for {Euclidean} {TSP} and other
  geometric problems.
\newblock In {\em 37th Conference on Foundations of Computer Science}, pages
  2--11, 1996.

\bibitem{BeI08}
Sergey Bereg and Hiro Ito.
\newblock Transforming graphs with the same degree sequence.
\newblock In {\em Computational Geometry and Graph Theory}, pages 25--32, 2008.

\bibitem{BeI17}
Sergey Bereg and Hiro Ito.
\newblock Transforming graphs with the same graphic sequence.
\newblock {\em Journal of Information Processing}, 25:627--633, 2017.

\bibitem{BMS19}
Ahmad Biniaz, Anil Maheshwari, and Michiel Smid.
\newblock Flip distance to some plane configurations.
\newblock {\em Computational Geometry}, 81:12--21, 2019.
\newblock URL: \url{https://arxiv.org/abs/1905.00791}.

\bibitem{BBH19}
Marthe Bonamy, Nicolas Bousquet, Marc Heinrich, Takehiro Ito, Yusuke Kobayashi,
  Arnaud Mary, Moritz M{\"{u}}hlenthaler, and Kunihiro Wasa.
\newblock The perfect matching reconfiguration problem.
\newblock In {\em 44th International Symposium on Mathematical Foundations of
  Computer Science}, volume 138 of {\em LIPIcs}, pages 80:1--80:14, 2019.

\bibitem{BoM16}
{\'{E}}douard Bonnet and Tillmann Miltzow.
\newblock Flip distance to a non-crossing perfect matching.
\newblock {\em arXiv}, 1601.05989, 2016.
\newblock URL: \url{http://arxiv.org/abs/1601.05989}.

\bibitem{BJ20}
Nicolas Bousquet and Alice Joffard.
\newblock Approximating shortest connected graph transformation for trees.
\newblock In {\em Theory and Practice of Computer Science}, pages 76--87, 2020.

\bibitem{BuKi22}
Maike Buchin and Bernhard Kilgus.
\newblock Fr{\'e}chet distance between two point sets.
\newblock {\em Computational Geometry}, 102:101842, 2022.

\bibitem{ChC19}
Miroslav Chleb{\'\i}k and Janka Chleb{\'\i}kov{\'a}.
\newblock Approximation hardness of {Travelling} {Salesman} via weighted
  amplifiers.
\newblock In {\em 25th International Computing and Combinatorics Conference},
  pages 115--127, 2019.

\bibitem{DDFGR22}
Arun~Kumar Das, Sandip Das, Guilherme~D. da~Fonseca, Yan Gerard, and Bastien
  Rivier.
\newblock Complexity results on untangling red-blue matchings.
\newblock {\em Computational Geometry}, 111:101974, 2023.
\newblock URL: \url{https://arxiv.org/abs/2202.11857}, \href
  {https://doi.org/10.1016/j.comgeo.2022.101974}
  {\path{doi:10.1016/j.comgeo.2022.101974}}.

\bibitem{Dav10}
Donald Davendra.
\newblock {\em Traveling salesman problem: Theory and applications}.
\newblock BoD--Books on Demand, 2010.

\bibitem{ERV14}
Matthias Englert, Heiko R{\"o}glin, and Berthold V{\"o}cking.
\newblock Worst case and probabilistic analysis of the {2-Opt} algorithm for
  the {TSP}.
\newblock {\em Algorithmica}, 68(1):190--264, 2014.
\newblock URL:
  \url{https://link.springer.com/content/pdf/10.1007/s00453-013-9801-4.pdf}.

\bibitem{ELS73}
Paul Erd{\"o}s, L{\'a}szl{\'o} Lov{\'a}sz, A~Simmons, and Ernst~G Straus.
\newblock Dissection graphs of planar point sets.
\newblock In {\em A survey of combinatorial theory}, pages 139--149. Elsevier,
  1973.

\bibitem{EKM13}
P{\'e}ter~L Erd{\H{o}}s, Zolt{\'a}n Kir{\'a}ly, and Istv{\'a}n Mikl{\'o}s.
\newblock On the swap-distances of different realizations of a graphical degree
  sequence.
\newblock {\em Combinatorics, Probability and Computing}, 22(3):366--383, 2013.

\bibitem{GuPu06}
Gregory Gutin and Abraham~P Punnen.
\newblock {\em The traveling salesman problem and its variations}, volume~12.
\newblock Springer Science \& Business Media, 2006.

\bibitem{Hak62}
Seifollah~Louis Hakimi.
\newblock On realizability of a set of integers as degrees of the vertices of a
  linear graph. {I}.
\newblock {\em Journal of the Society for Industrial and Applied Mathematics},
  10(3):496--506, 1962.

\bibitem{Hak63}
Seifollah~Louis Hakimi.
\newblock On realizability of a set of integers as degrees of the vertices of a
  linear graph {II}. uniqueness.
\newblock {\em Journal of the Society for Industrial and Applied Mathematics},
  11(1):135--147, 1963.

\bibitem{HNU99}
Ferran Hurtado, Marc Noy, and Jorge Urrutia.
\newblock Flipping edges in triangulations.
\newblock {\em Discrete \& Computational Geometry}, 22(3):333--346, 1999.

\bibitem{phdJof}
Alice Joffard.
\newblock {\em Graph domination and reconfiguration problems}.
\newblock PhD thesis, Université Claude Bernard Lyon 1, 2020.

\bibitem{Law72}
Charles~L Lawson.
\newblock Transforming triangulations.
\newblock {\em Discrete Mathematics}, 3(4):365--372, 1972.

\bibitem{LuP15}
Anna Lubiw and Vinayak Pathak.
\newblock Flip distance between two triangulations of a point set is
  {NP}-complete.
\newblock {\em Computational Geometry}, 49:17--23, 2015.

\bibitem{Mit99}
Joseph~SB Mitchell.
\newblock Guillotine subdivisions approximate polygonal subdivisions: A simple
  polynomial-time approximation scheme for geometric {TSP}, k-{MST}, and
  related problems.
\newblock {\em SIAM Journal on computing}, 28(4):1298--1309, 1999.

\bibitem{NiN18}
Naomi Nishimura.
\newblock Introduction to reconfiguration.
\newblock {\em Algorithms}, 11(4), 2018.

\bibitem{OdW07}
Yoshiaki Oda and Mamoru Watanabe.
\newblock The number of flips required to obtain non-crossing convex cycles.
\newblock In {\em Kyoto International Conference on Computational Geometry and
  Graph Theory}, pages 155--165, 2007.

\bibitem{Pil14}
Alexander Pilz.
\newblock Flip distance between triangulations of a planar point set is
  apx-hard.
\newblock {\em Computational Geometry}, 47(5):589--604, 2014.

\bibitem{RaSm98}
Satish~B Rao and Warren~D Smith.
\newblock Approximating geometrical graphs via “spanners” and
  “banyans”.
\newblock In {\em Proceedings of the thirtieth annual ACM symposium on Theory
  of computing}, pages 540--550, 1998.

\bibitem{Heu13}
Jan van~den Heuvel.
\newblock The complexity of change.
\newblock {\em Surveys in Combinatorics}, 409:127--160, 2013.

\bibitem{VLe81}
Jan van Leeuwen.
\newblock Untangling a traveling salesman tour in the plane.
\newblock In {\em 7th Workshop on Graph-Theoretic Concepts in Computer
  Science}, 1981.

\bibitem{Wil99}
Todd~G Will.
\newblock Switching distance between graphs with the same degrees.
\newblock {\em SIAM Journal on Discrete Mathematics}, 12(3):298--306, 1999.

\bibitem{WCL09}
Ro{-}Yu Wu, Jou{-}Ming Chang, and Jia{-}Huei Lin.
\newblock On the maximum switching number to obtain non-crossing convex cycles.
\newblock In {\em 26th Workshop on Combinatorial Mathematics and Computation
  Theory}, pages 266--273, 2009.

\end{thebibliography}
\end{document}